\documentclass[a4paper,12pt]{article}
\usepackage{amsmath,amsthm,amssymb}
\usepackage{amsmath}
\usepackage{amsfonts}

\advance \textheight by \topskip
\advance \textheight by 1pt
\makeatother

\def\b{\begin{equation}}
\def\e{\end{equation}}

\numberwithin{equation}{section}

\newtheorem{theorem}{Theorem}[section]
\newtheorem{prop}[theorem]{Proposition}
\newtheorem{lem}[theorem]{Lemma}
\newtheorem{cor}[theorem]{Corollary}
{\theoremstyle{definition}
\newtheorem{definition}[theorem]{Definition}
\newtheorem{rem}[theorem]{Remark}
\newtheorem{ex}[theorem]{Example}}

\def\wt{{\rm wt}\,}

\def\b{\begin{equation}}
\def\e{\end{equation}}

\title{Hyperelliptic sigma functions and the Kadomtsev-Petviashvili equation\footnote{This is the accepted version of an article published in Physica D: Nonlinear Phenomena. 
The final published version is available at: https://doi.org/10.1016/j.physd.2025.134819}}
\author{Takanori Ayano\footnote{Osaka Central Advanced Mathematical Institute, Osaka Metropolitan University, Osaka, Japan. \newline \hspace{3ex} Email: ayano@omu.ac.jp} \hspace{1ex} and \hspace{1ex} Victor M. Buchstaber\footnote{Steklov Mathematical Institute of Russian Academy of Sciences, Moscow, Russia. \newline \hspace{3ex} Email: buchstab@mi-ras.ru
\newline \hspace{3ex} Key words: hyperelliptic sigma function, hyperelliptic function, KP equation. 
\newline \hspace{3ex} MSC classes: 14H42, 14K25, 14H70, 14H81.}}

\date{}

\begin{document}
\maketitle

\begin{flushright}
\textit{Dedicated to the memory of Vladimir Evgenievich Zakharov.}
\end{flushright}

\begin{abstract}
In this paper, a theory of hyperelliptic functions based on multidimensional sigma functions is developed and explicit formulas for hyperelliptic solutions to the Kadomtsev-Petviashvili equations KP-I and KP-II are obtained. 
The long-standing problem of describing the dependence of these solutions on the variation of the coefficients of the defining equation of a hyperelliptic curve, which are integrals of the equations, is solved.
\end{abstract}

\section{Introduction}

The Kadomtsev-Petviashvili equation (1970) is one of the most famous (2+1)-equations in the theory of nonlinear waves (cf. \cite{Kadomtsev-Petviashvili-1970}).  
It is a natural generalization of the (1+1)-Korteweg-de Vries equation (1895). 
In \cite{Zakharov-Shabat-1974} and \cite{Zakharov-Shabat-1979}, Zakharov and Shabat described solutions to the KP equation with the condition that they decrease rapidly at infinity. 
For a discussion of the integrability of the KP-II equation and non-integrability of the KP-I equation, see \cite{Zakharov-Schulman-1988}. 
In \cite{Novikov}, Sergei Novikov developed the theory of finite-zone integration of equations of mathematical physics. 
In \cite{Novikov}, it was noted that ``Our work is based on certain simple but fundamental algebraic properties of equations
admitting the Lax representation which are strongly degenerate in the problem with rapidly decreasing
functions (for $x\to\pm\infty$), and have therefore not been noted. Finally, it is essential to note the nonlinear
``superposition law for waves" for the KdV equation which in the periodic case has an interesting algebraic-geometric
interpretation."
In \cite{Its-Matveev-1975}, Its and Matveev constructed a solution to the KdV equation in terms of the Riemann theta function defined by the lattice of periods of holomorphic differentials on an algebraic curve of arbitrary genus 
and winding vectors defined by periods of abelian differentials of the second kind. 
Krichever added one direction vector in the theta functional solution of Its and Matveev to the KdV equation and constructed a solution to the KP equation (cf. \cite{Harnad-Enolski-2011}, \cite{Krichever-1977}). 
Krichever's result on solutions to the KP equation led Novikov to the famous conjecture, an approach to solve the Riemann-Schottky problem in terms of the KP equation, \cite{Krichever-2008}, which was solved by T. Shiota \cite{Shiota-1986}. 
In \cite{Zhao-Fan-Luo-2016}, solutions to the KP equation in terms of the Riemann theta function associated with a hyperelliptic curve were considered. 
In our work, explicit formulas for solutions to the Kadomtsev-Petviashvili equations KP-I and KP-II are obtained in terms of hyperelliptic functions. 
The hyperelliptic functions associated with a hyperelliptic curve with one infinite point are defined by the logarithmic derivatives of the sigma function associated with the curve. 
This hyperelliptic sigma function is completely determined by the coefficients of the defining equation of the curve. 
The hyperelliptic functions associated with a hyperelliptic curve with two infinite points are defined by using the Abel-Jacobi map. 
The hyperelliptic functions associated with the hyperelliptic curve with two infinite points are determined by the coefficients of the defining equation of the curve and a branch point of the curve.  
We have not only obtained an explicit form of solutions but also related the coefficients of the defining equation of the curve to physical parameters.  
Our solutions require scaling, i.e., multiplication by scalars of variables and functions, and linear transformations of variables of the ``traveling wave" type. 
To describe the behavior of the sigma function with shifts by periods, its expression through the theta function is used (see Proposition \ref{2025.2.23.18765432042224455}). 
The remarkable fact is that the sigma function associated with a hyperelliptic curve with one infinite point is a solution to the system of multidimensional heat equations in a nonholonomic frame, 
which is completely determined by the coefficients of the defining equation of the curve (cf. \cite{BL-2004-Heat-Equations}). 
In this case, our solutions to the KP equation are determined only by the coefficients of the defining equation of the curve and the well-known problem of constructing real-valued solutions is solved.

For a positive integer $g$, let $C$ be the hyperelliptic curve of genus $g$ defined by 
\begin{equation}
Y^2=X^{2g+1}+\lambda_2X^{2g}+\lambda_4X^{2g-1}+\cdots+\lambda_{4g}X+\lambda_{4g+2}, \qquad \lambda_i\in\mathbb{C}.\label{2025.2.3.0897634143}
\end{equation}
We assign weights for $X$, $Y$, and $\lambda_i$ as $\wt(X)=2$, $\wt(Y)=2g+1$, and $\wt(\lambda_i)=i$.  
The equation (\ref{2025.2.3.0897634143}) has the homogeneous weight $4g+2$ with respect to the coefficients $\lambda_i$ and the variables $X,Y$.  
Let $V$ be the hyperelliptic curve of genus $g$ defined by 
\begin{equation}
y^2=\nu_{0}x^{2g+2}+\nu_{2}x^{2g+1}+\cdots+\nu_{4g+2}x+\nu_{4g+4}, \quad \nu_i\in\mathbb{C},\quad\nu_{0}\neq0.\label{2025.2.3.08976341433461200}
\end{equation}
We assign weights for $x$, $y$, and $\nu_i$ as $\wt(x)=2$, $\wt(y)=2g+2$, and $\wt(\nu_i)=i$.  
The equation (\ref{2025.2.3.08976341433461200}) has the homogeneous weight $4g+4$ with respect to the coefficients $\nu_i$ and the variables $x,y$. 
We consider the hyperelliptic functions $\wp_{i,j}(u_1,u_3,\dots,u_{2g-1})$ with $i,j=1,3,\dots,2g-1$ associated with the curve $C$, which are meromorphic functions on $\mathbb{C}^g$, 
and assign weights for $u_i$ and $\wp_{i,j}$ as $\wt(u_i)=-i$ and $\wt(\wp_{i,j})=i+j$.  
We consider the hyperelliptic functions $\mathcal{P}_{i,j}(v_{2g},v_{2g-2},\dots,v_2)$ with $i,j=2,4,\dots,2g$ associated with the curve $V$, which are meromorphic functions on $\mathbb{C}^g$, 
and assign weights for $v_i$ and $\mathcal{P}_{i,j}$ as $\wt(v_i)=-i$ and $\wt(\mathcal{P}_{i,j})=i+j$.  
Let us describe our solutions to the KP equation. 
We consider the case $g\ge3$, assume $\lambda_{4g+2}\neq0$, and for $g\ge4$ take constants $b_i\in\mathbb{C}$ with $1\le i\le g-3$. 
Let 
\[\varphi(t_1,t_2,t_3)=-2\wp_{2g-1,2g-1}\left(b_1,\dots,b_{g-3}, \mathfrak{c}t_3, \mathfrak{d}t_2, t_1+\mathfrak{e}t_2\right)-\mathfrak{f},\]
where 
\[\mathfrak{c}=-16\lambda_{4g+2},\quad \mathfrak{d}=2\sqrt{-3\lambda_{4g+2}},\quad \mathfrak{e}=\frac{\lambda_{4g}}{\sqrt{-3\lambda_{4g+2}}},\quad \mathfrak{f}=\frac{2}{3}\lambda_{4g-2}+\frac{\lambda_{4g}^2}{18\lambda_{4g+2}}.\]
For the curve $C$, we assign weights for $t_i$ with $i=1,2,3$ as $\wt(t_i)=(1-2g)i$. We have $\wt(\varphi)=4g-2$.  
In Corollary \ref{2024.12.6.1}, we prove that the function $\varphi$ satisfies the KP-$\mathrm{I}$ equation
\[\partial_{t_1}(\partial_{t_3}\varphi+6\varphi\partial_{t_1}\varphi+\partial_{t_1}^3\varphi)=\partial_{t_2}^2\varphi,\]
where $\partial_{t_i}=\partial/\partial t_i$. 
In \cite{BEL-2000}, it was pointed out that if $g\ge3$, under certain restrictions on the coefficients of the defining equation of the curve, the function $\wp_{2g-1,2g-1}$ is a solution to the KP equation. 
In Corollary \ref{2024.12.6.1} of our paper, we give a simple explicit condition on the coefficients of the defining equation of the curve and under this condition we prove that the function $\wp_{2g-1,2g-1}$ is a solution to the KP equation. 
We consider the case $g\ge3$ and for $g\ge4$ take constants $c_i\in\mathbb{C}$ with $1\le i\le g-3$. 
Let 
\[\psi(t_1,t_2,t_3)=-2\mathcal{P}_{2,2}\left(c_1,\dots,c_{g-3},\alpha t_3, \beta t_2, t_1+\gamma t_2\right)-\delta,\]
where 
\[\alpha=-16\nu_{0},\quad \beta=2\sqrt{-3\nu_{0}},\quad \gamma=\frac{\nu_{2}}{\sqrt{-3\nu_{0}}},\quad \delta=\frac{2}{3}\nu_{4}+\frac{\nu_{2}^2}{18\nu_{0}}.\]
For the curve $V$, we assign weights for $t_i$ with $i=1,2,3$ as $\wt(t_i)=-2i$. We have $\wt(\psi)=4$. 
We derive the differential relations between the hyperelliptic functions $\mathcal{P}_{i,j}$ explicitly for any $g$ (see Theorem \ref{2024.10.22.1111}). 
In Corollary \ref{4}, by using these differential relations, we prove that the function $\psi$ satisfies the KP-$\mathrm{I}$ equation
\[\partial_{t_1}(\partial_{t_3}\psi+6\psi\partial_{t_1}\psi+\partial_{t_1}^3\psi)=\partial_{t_2}^2\psi.\]
Our solutions to the KP equation use not only differentiation operators but also argument shift operators.  
Our solutions to the KP equation are consistent with the grading. 
Grading is the fundamental difference between the sigma function and the theta function.
The theta function does not allow grading since its arguments are normalized. 
We can obtain solutions to the KP-II equation (see Remark \ref{2025.1.30.1}). 

In \cite{Kl1} and \cite{Kl2}, Klein generalized the Weierstrass elliptic sigma function to the multidimensional sigma functions associated with hyperelliptic curves. 
On this problem, Klein published 3 works (1886--1890). 
Pay attention to the papers \cite{B-1898} and \cite{B} by Baker. 
In 1923, a 3-volume collection of Klein's scientific works was published. 
There is no doubt that Klein knew Baker's results. 
However, in this collection Klein emphasized that the theory of hyperelliptic functions was still far from complete. 
Klein and Baker did not discuss the equations of mathematical physics. 
The development of the theory of multidimensional sigma functions in the direction of applications to problems of mathematical physics began with the works of Buchstaber, Enolski, and Leykin (cf. \cite{BEL-97-1}, \cite{BEL-97-2}, \cite{BEL-2012}, \cite{BEL-2018}).  
Over the past 30 years, a number of authors have successfully joined in the development of the classical results of Klein and Baker with applications in mathematical physics.

In \cite{B}, Baker derived a fundamental formula on the hyperelliptic functions associated with the curve $C$. 
Baker did not set the problem of describing all relations in the field of meromorphic functions on the Jacobian variety of the hyperelliptic curve. 
In \cite{BEL-97-1}, this problem was set and the generators of the ideal of relations between the hyperelliptic functions associated with the curve $C$ were described explicitly.  
The solution to this problem made it possible to specify the equations describing the Jacobian variety of the hyperelliptic curve explicitly. 
It is remarkable that these generators explicitly contain the KdV equation (cf. \cite{BEL-97-1}). 
This ideal also contains the KP equation. 
In \cite{Buchstaber-Mikhailov-2021}, Buchstaber and Mikhailov described the connection between the hyperelliptic functions based on the sigma functions and the known constructions of the KdV hierarchy.
In \cite{Yang}, in the case $g=3$, a solution to the KP equation in terms of the hyperelliptic function associated with the curve $C$ was considered.

In \cite{B}, Baker introduced basic hyperelliptic functions associated with the curve $V$ (see Definition \ref{2024.11.29.12345}) and derived a fundamental formula on these functions (see Lemma \ref{2024.10.21.111}).  
Further, in \cite{B}, the differential relations between the hyperelliptic functions associated with the curve $V$ were described explicitly for $g=1,2,3$. 
In \cite{B}, Baker used the Abel-Jacobi map and did not actually need the sigma function. 
In \cite{M}, in the case $g=3$, it was proved that the hyperelliptic function associated with the curve $V$ satisfies the KP equation.  
Above we described the solution to the KP equation of weight $4g-2$ for the curve $C$ and the solution to the KP equation of weight $4$ for the curve $V$. 
We also note the following results of our paper. 
\begin{itemize}

\item We construct a solution to the KP equation of weight $2$ for the curve $C$ (see Proposition \ref{2025.2.6.19561456}). 

\item We describe the relations between the hyperelliptic functions associated with the curves $C$ and $V$ explicitly (see Proposition \ref{2024.11.30.9065376}). 
\end{itemize}

In the works of Buchstaber, Enolski, and Leykin, the focus was on a hyperelliptic curve with one infinite point. 
In our paper, we also obtain results for a hyperelliptic curve with two infinite points. 
The problem to find hyperelliptic solutions to the KP equation is non-trivial. 
In the general approach of the theory of finite-zone integration, only hyperelliptic solutions to the KdV equation were discussed. 
This problem was posed in the works of Buchstaber, Enolski, and Leykin on the basis of the analysis of differential relations in the field of hyperelliptic functions.
We have obtained an interesting and new relationship of results for the curves $C$ and $V$. Functions admit two types of operations: 
\begin{enumerate}
\item differentiation with respect to one of the variables,

\item linear transformation of variables.
\end{enumerate}
For the abelian functions, in the first case the lattice of periods does not change and in the second case it is linearly transformed. 
For the curve $C$, we use a linear transformation of the arguments (see Remark \ref{2025.1.23.1234}). 
For the curve $V$, this is not required. 


An $n$-dimensional theta function is an entire function of $n$ variables whose power series expansion is given in terms of a special lattice $\Gamma$ of rank $2n$. 
In the case where an abelian variety $\mathbb{C}^n/\Gamma$ is the Jacobian variety of an algebraic curve, the lattice $\Gamma$ is given by the periods of holomorphic differentials on this curve. 
In \cite{Ayano-Nakayashiki-2013} and \cite{Nakayashiki-2010}, for $(n,s)$ curves and telescopic curves, the relations between the tau functions and the sigma functions were studied and solutions to the KP hierarchy in terms of the tau functions were described. 
In \cite{Kodama-2018} and \cite{Kodama-Xie-2021}, Kodama described solutions to the KP problem based on Hirota's approach using the tau function and a description of solutions in terms of Grassmann manifolds. 
In \cite{Khan-Akbar-2015}, \cite{Kumar-Tiwari-Kumar-2017}, \cite{C-Liu-2010}, \cite{L.-Wei-2011}, and \cite{Xu-Chen-Dai-2014}, solutions to the KP equation in terms of hyperbolic functions, trigonometric functions, exponential functions, and 
logarithmic functions were constructed.  
Our paper describes solutions to the KP equation in terms of multidimensional sigma functions. 
Unlike theta functions, the multidimensional sigma function is an entire function whose coefficients of the power series expansion are polynomials in the coefficients of the defining equation of the curve. 
The coefficients of the defining equation of the curve are integrals of the KP equation.  
By studying the deformations of the coefficients of the defining equation of the curve, we study the deformation of the integrals. 
Describing the dependence of the solution on the deformation of the integrals is an important task. 
This problem can be solved by studying the solutions to the KP equation in terms of the multidimensional sigma functions.



In \cite{A-E-E-2003}, \cite{A-E-E-2004}, and \cite{B-1907}, the identities for hyperelliptic functions of genus 2 which are different from the hyperelliptic functions considered in our paper were studied. 

The present paper is organized as follows. 
In Section 2, we review the definition and properties of the hyperelliptic sigma functions.  
In Section 3, we construct solutions to the KP equation in terms of the hyperelliptic functions associated with a hyperelliptic curve with one infinite point.  
In Section 4, we review the definition of the hyperelliptic functions associated with a hyperelliptic curve with two infinite points, which was given in \cite{B}. 
In Section 5, we derive the differential relations between the hyperelliptic functions associated with the hyperelliptic curve with two infinite points and prove that 
the hyperelliptic function associated with the hyperelliptic curve with two infinite points satisfies the KP equation for any genus. 
In Section 6, we consider the inversion problem of the Abel-Jacobi map for the hyperelliptic curve with two infinite points. 
In Section 7, we describe the relations between the hyperelliptic functions associated with the hyperelliptic curve with one infinite point and the hyperelliptic functions associated with the hyperelliptic curve with two infinite points explicitly.

\section{Hyperelliptic sigma functions}\label{2024.10.21.1}

In this section, we review the definition of the hyperelliptic sigma functions and give facts about them which will be used later on.  
For details of the hyperelliptic sigma functions, see \cite{BEL-97-1}, \cite{BEL-97-2}, and \cite{BEL-2012}.  

For a positive integer $g$, let us consider the polynomial in $X$ 
\[
M(X)=X^{2g+1}+\lambda_2X^{2g}+\lambda_4X^{2g-1}+\cdots+\lambda_{4g}X+\lambda_{4g+2}, \qquad \lambda_i\in\mathbb{C}.
\]
We assume that $M(X)$ has no multiple roots and consider the non-singular hyperelliptic curve of genus $g$
\[C=\Bigl\{(X,Y)\in\mathbb{C}^2 \Bigm|Y^2=M(X)\Bigr\}.\]
We assign weights for $X$, $Y$, and $\lambda_i$ as $\wt(X)=2$, $\wt(Y)=2g+1$, and $\wt(\lambda_i)=i$.  
The equation $Y^2=M(X)$ has the homogeneous weight $4g+2$ with respect to the coefficients $\lambda_i$ and the variables $X,Y$.  
A basis of the vector space consisting of holomorphic 1-forms on $C$ is given by
\[
\omega_i=-\frac{X^{g-i}}{2Y}dX, \qquad 1\le i\le g. 
\]
We set $\omega={}^t(\omega_1,\dots,\omega_g)$. 
Let us consider the following meromorphic 1-forms on $C$:  
\begin{equation}
\eta_i=-\frac{1}{2Y}\sum_{k=g-i+1}^{g+i-1}(k+i-g)\lambda_{2g+2i-2k-2}X^kdX, \qquad 1\le i\le g,\label{2024.10.21.345}
\end{equation}
which are holomorphic at any point except $\infty$. 
In (\ref{2024.10.21.345}), we set $\lambda_0=1$. 
Let $\{A_i,B_i\}_{i=1}^g$ be a canonical basis in the one-dimensional homology group of the curve $C$. 
We define the period matrices by 
\[
2\omega'=\left(\int_{A_j}\omega_i\right),\quad
2\omega''=\left(\int_{B_j}\omega_i\right),\quad 
-2\eta'=\left(\int_{A_j}\eta_i\right),
\quad
-2\eta''=\left(\int_{B_j}\eta_i\right).
\]
We define the lattice of periods $\Lambda=\bigl\{2\omega'm_1+2\omega''m_2\mid m_1,m_2\in\mathbb{Z}^g\bigr\}$ and consider the Jacobian variety $\operatorname{Jac}(C)=\mathbb{C}^g/\Lambda$. 
The normalized period matrix is given by $\tau=(\omega')^{-1}\omega''$. 
Let $\tau\delta'+\delta''$ with $\delta',\delta''\in\mathbb{R}^g$ be the Riemann constant with respect to $(\{A_i, B_i\}_{i=1}^g,\infty)$. We denote the imaginary unit by $\textbf{i}$.
The sigma function $\sigma(u)$ associated with the curve $C$, $u={}^t(u_1, u_3, \dots, u_{2g-1})\in\mathbb{C}^g$, is defined by
\[
\sigma(u)=\varepsilon\exp\biggl(\frac{1}{2}{}^tu\eta'(\omega')^{-1}u\biggr)\theta\begin{bmatrix}\delta'\\ \delta'' \end{bmatrix}\bigl((2\omega')^{-1}u,\tau\bigr),
\]
where $\theta\begin{bmatrix}\delta'\\ \delta'' \end{bmatrix}(u,\tau)$ is the Riemann theta function with the characteristics $\begin{bmatrix}\delta'\\ \delta'' \end{bmatrix}$ defined by
\[
\theta\begin{bmatrix}\delta'\\ \delta'' \end{bmatrix}(u,\tau)=\sum_{n\in\mathbb{Z}^g}\exp\bigl\{\pi{\bf i}\,{}^t(n+\delta')\tau(n+\delta')+2\pi{\bf i}\,{}^t(n+\delta')(u+\delta'')\bigr\}
\]
and $\varepsilon$ is a non-zero constant. 
The characteristics of this sigma function correspond to the vector of the Riemann constant. 

\begin{prop}[{\cite[pp.~7, 8]{BEL-97-1}}]\label{2025.2.23.18765432042224455}
For $m_1,m_2\in\mathbb{Z}^g$, let $\Omega=2\omega'm_1+2\omega''m_2$.
Then, for $u\in\mathbb{C}^g$, we have
\[
\sigma(u+\Omega)/\sigma(u) =(-1)^{2({}^t\delta'm_1-{}^t\delta''m_2)+{}^tm_1m_2}\exp\bigl\{{}^t(2\eta'm_1+
2\eta''m_2)(u+\omega'm_1+\omega''m_2)\bigr\}.
\]
\end{prop}

\section{Solutions to the KP equation in terms of the hyperelliptic functions associated with the curve $C$}

For an integer $k\ge2$, let $\wp_{i_1,\dots,i_k}=-\partial_{u_{i_1}}\cdots\;\partial_{u_{i_k}}\log\sigma$, where $\partial_{u_l}=\partial/\partial u_l$. 
We assign weights for $\wp_{i_1,\dots,i_k}$ as $\wt(\wp_{i_1,\dots,i_k})=i_1+\cdots+i_k$.  
For $1\le i\le g$, we take points $S_i=(X_i,Y_i)\in C\setminus\{\infty\}$ such that $S_i\neq\tau_1(S_j)$ if $i\neq j$, where $\tau_1$ is the hyperelliptic involution of $C$
\[\tau_1: C\to C,\qquad (X,Y)\mapsto(X,-Y).\] 
Let 
\[u=\sum_{i=1}^g\int_{\infty}^{S_i}\omega.\]
For $1\le k\le g$, let $h_k(z_1,\dots,z_g)$ be the elementary symmetric polynomial of degree $k$ in $g$ variables $z_1,\dots,z_g$. 
We set $h_0(z_1,\dots,z_g)=1$. 
For non-negative integers $i$ and $j$, let $\delta_{i,j}$ be the Kronecker delta
\[\delta_{i,j}=\left\{\begin{array}{ll}1& \mathrm{if}\;\;i=j\\0&\mathrm{if}\;\;i\neq j\end{array}\right..\] 

\begin{theorem}[{\cite[Theorem 2.2]{BEL-97-1}}]\label{2024.10.22.1}
We have 
\[h_k(X_1,\dots,X_g)=(-1)^{k-1}\wp_{1, 2k-1}(u),\qquad 1\le k\le g.\]
\end{theorem}

\begin{rem}
In \cite[pp.~155, 156]{B}, for $g=1,2,3$, the functions $\wp_{i,j,k,l}$ with $i,j,k,l=1,3,\dots,2g-1$ were expressed in terms of $\wp_{m,n}$ explicitly. 
\end{rem}

\begin{rem}
In \cite[pp.~57, 58, 88]{BEL-97-2}, for $g=2,3$, the functions $\wp_{i,j,k,l}$ with $i,j,k,l=1,3,\dots,2g-1$ were expressed in terms of $\wp_{m,n}$ explicitly. 
\end{rem}

\begin{rem}
In \cite{HST}, for $g=4$, the functions $\wp_{i,j,k,l}$ with $i,j,k,l=1,3,5,7$ were expressed in terms of $\wp_{m,n}$ explicitly. 
\end{rem}


If the number of $i$ is $d$ in the suffixes of $\wp_{i_1,\dots,i_k}$, then we use the notation $i\cdot d$. 
For example, we denote $\wp_{1,1}$, $\wp_{1,1,3}$, and $\wp_{1,1,3,3}$ by $\wp_{1\cdot2}$, $\wp_{1\cdot2,3}$, and $\wp_{1\cdot2, 3\cdot2}$, respectively. 
If $i_j\notin\{1,3,\dots,2g-1\}$ for some $1\le j\le k$, we set $\wp_{i_1,\dots,i_k}=0$.

\begin{theorem}[{\cite[Corollaries 3.1.2, 3.1.3, Theorem 3.2]{BEL-97-1}}]\label{2024.11.20.1}
$(\mathrm{i})$ For $1\le i\le g$, the following equalities hold$:$  
\[\wp_{1\cdot3,2i-1}=(6\wp_{1\cdot2}+4\lambda_2)\wp_{1,2i-1}+6\wp_{1,2i+1}-2\wp_{3,2i-1}+2\delta_{1,i}\lambda_4.\]
$(\mathrm{ii})$ For $1\le i,j\le g$, the following equalities hold$:$ 
\begin{align*}
\wp_{1\cdot2,2i-1}\wp_{1\cdot2,2j-1}&=4\wp_{1\cdot2}\wp_{1,2i-1}\wp_{1,2j-1}-2(\wp_{1,2i-1}\wp_{3,2j-1}+\wp_{1,2j-1}\wp_{3,2i-1})\\
&\quad +4(\wp_{1,2j-1}\wp_{1,2i+1}+\wp_{1,2i-1}\wp_{1,2j+1})+4\wp_{2i+1,2j+1}\\
&\quad -2(\wp_{2j-1,2i+3}+\wp_{2i-1,2j+3})+4\lambda_2\wp_{1,2i-1}\wp_{1,2j-1}\\
&\quad +2\lambda_4(\delta_{i,1}\wp_{1,2j-1}+\delta_{j,1}\wp_{1,2i-1})+4\lambda_{4i+2}\delta_{i,j}+2(\lambda_{4i}\delta_{i,j+1}+\lambda_{4j}\delta_{i+1,j}).
\end{align*}
$(\mathrm{iii})$ For $1\le i,j\le g$, the following equalities hold$:$ 
\[\wp_{1\cdot2,2j-1}\wp_{1,2i-1}-\wp_{1\cdot2,2i-1}\wp_{1,2j-1}+\wp_{1,2i+1,2j-1}-\wp_{1,2i-1,2j+1}=0.\]
$(\mathrm{iv})$ For $1\le i,j\le g$, the following equalities hold$:$ 
\[\wp_{1\cdot3,2j-1}\wp_{1,2i-1}-\wp_{1\cdot3,2i-1}\wp_{1,2j-1}+\wp_{1\cdot2,2i+1,2j-1}-\wp_{1\cdot2,2i-1,2j+1}=0.\]
\end{theorem}

\begin{rem}
Theorem \ref{2024.11.20.1} was obtained in \cite{BEL-97-1}. 
We changed the suffixes of $\wp_{i_1,\dots,i_k}$ in \cite{BEL-97-1} to use the grading. 
The suffix $g$ in \cite{BEL-97-1} is replaced with $1$. 
\end{rem}

\begin{rem}
The equalities in Theorem \ref{2024.11.20.1} (i)--(iv) have the homogeneous weights $2i+2$, $2i+2j+2$, $2i+2j+1$, and $2i+2j+2$, respectively.  
\end{rem}

\begin{rem}
We consider the case $g\ge2$ and for $g\ge3$ take constants $d_i\in\mathbb{C}$ with $3\le i\le g$. 
In \cite[Theorem 4.12]{BEL-97-1}, by using the equality in Theorem \ref{2024.11.20.1} (i) with $i=1$, it was proved that $\mathcal{G}(t_1,t_3)=2\wp_{1\cdot2}(t_1,t_3,d_3,\dots,d_g)+2\lambda_2/3$ satisfies the KdV equation
\[4\partial_{t_3}\mathcal{G}+6\mathcal{G}\partial_{t_1}\mathcal{G}-\partial_{t_1}^3\mathcal{G}=0.\]
\end{rem}

We consider the case $g\ge2$ and for $g\ge3$ take constants $\varrho_i\in\mathbb{C}$ with $3\le i\le g$. 
Let us consider the function
\[\Upsilon(t_1,t_2,t_3)=-2\wp_{1\cdot2}(t_1+2\sqrt{\lambda_2}t_2, -4t_3, \varrho_3,\dots,\varrho_g).\]

\begin{prop}\label{2025.2.6.19561456}
The function $\Upsilon$ satisfies the KP-$\mathrm{I}$ equation
\[\partial_{t_1}(\partial_{t_3}\Upsilon+6\Upsilon\partial_{t_1}\Upsilon+\partial_{t_1}^3\Upsilon)=\partial_{t_2}^2\Upsilon.\]
\end{prop}

\begin{proof}
From Theorem \ref{2024.11.20.1} (i) with $i=1$, we have 
\[\wp_{1\cdot4}=(6\wp_{1\cdot2}+4\lambda_2)\wp_{1\cdot2}+4\wp_{1,3}+2\lambda_4.\]
By differentiating this equality with respect to $u_1$ twice, we obtain
\[\wp_{1\cdot6}=12\wp_{1\cdot2}\wp_{1\cdot4}+12\wp_{1\cdot3}^2+4\lambda_2\wp_{1\cdot4}+4\wp_{1\cdot3,3}.\]
We have 
\begin{gather*}
\partial_{t_1}\partial_{t_3}\Upsilon=8\wp_{1\cdot3,3},\quad (\partial_{t_1}\Upsilon)^2=4\wp_{1\cdot3}^2,\quad \Upsilon\partial_{t_1}^2\Upsilon=4\wp_{1\cdot2}\wp_{1\cdot4},\\
\partial_{t_1}^4\Upsilon=-2\wp_{1\cdot6},\quad \partial_{t_2}^2\Upsilon=-8\lambda_2\wp_{1\cdot4}.
\end{gather*}
From the equalities above, we obtain the statement of the proposition. 
\end{proof}

\begin{rem}
Proposition \ref{2025.2.6.19561456} is an analog of the following well-known result. 
For $g=1$, the function $\mathcal{H}(t_1,t_3)=2\wp_{1\cdot2}(t_1+\lambda_2t_3)$ satisfies the KdV equation
\[4\partial_{t_3}\mathcal{H}+6\mathcal{H}\partial_{t_1}\mathcal{H}-\partial_{t_1}^3\mathcal{H}=0.\]
\end{rem}


The statement of the following theorem was formulated in the conclusion in \cite[p.~156]{B}. 
Since it is used to derive the KP equation from Theorem \ref{2024.11.20.1}, we give its proof. 

\begin{theorem}\label{2024.10.30.1}
For any $g\ge1$, the following equality holds$:$
\begin{align*}
\wp_{(2g-1)\cdot4}&=6\wp_{(2g-1)\cdot2}^2+4\lambda_{4g}\wp_{2g-1,2g-3}+4\lambda_{4g+2}(4\wp_{2g-1,2g-5}-3\wp_{(2g-3)\cdot2})\\
&\quad +4\lambda_{4g-2}\wp_{(2g-1)\cdot2}-8\lambda_{4g+2}\lambda_{4g-6}+2\lambda_{4g}\lambda_{4g-4}, 
\end{align*}
where we set $\lambda_0=1$ and $\lambda_i=0$ for $i<0$. 
\end{theorem}

\begin{proof}
For $g=1$, this theorem is well known. 
For $g=2,3$, this theorem was given in \cite[pp.~58, 88]{BEL-97-2}. 
For $g\ge4$, we prove this theorem from Theorem \ref{2024.11.20.1} directly. 
For a positive integer $d$, we set 
\[p_d=\wp_{1\cdot d,2g-1},\qquad q_d=\wp_{1\cdot d,2g-3},\qquad r_d=\wp_{1\cdot d,2g-5},\qquad s_d=\wp_{1\cdot d,2g-7}.\]
For a meromorphic function $H(u_1,u_3,\dots,u_{2g-1})$ on $\mathbb{C}^g$, we use the notation $\dot{H}=\partial_{u_{2g-1}}H$. 
From Theorem \ref{2024.11.20.1} (ii) with $i=j=g-1$, we have 
\begin{equation}
\wp_{(2g-1)\cdot2}=\frac{1}{4}q_2^2-q_1^2(\wp_{1\cdot2}+\lambda_2)+q_1(\wp_{3,2g-3}-2p_1)-\lambda_{4g-2}.\label{2024.11.20.222}
\end{equation}
From Theorem \ref{2024.11.20.1} (i) with $i=g-1$, we have 
\begin{equation}
\wp_{3,2g-3}=q_1(3\wp_{1\cdot2}+2\lambda_2)+3p_1-\frac{1}{2}q_3.\label{2024.11.21.1}
\end{equation}
By substituting (\ref{2024.11.21.1}) into (\ref{2024.11.20.222}), we obtain
\[\wp_{(2g-1)\cdot2}=\frac{1}{4}q_2^2+q_1^2(2\wp_{1\cdot2}+\lambda_2)+q_1\left(p_1-\frac{1}{2}q_3\right)-\lambda_{4g-2}.\]
From Theorem \ref{2024.11.20.1} (iii) and (iv), we have 
\begin{gather*}
\dot{p}_1=q_2p_1-p_2q_1,\qquad \dot{q}_1=r_2p_1-p_2r_1,\qquad \dot{r}_1=s_2p_1-p_2s_1,\\
\dot{p}_2=q_3p_1-p_3q_1,\qquad \dot{q}_2=r_3p_1-p_3r_1,\qquad \dot{r}_2=s_3p_1-p_3s_1.
\end{gather*}
We have 
\[\dot{q}_3=\partial_{u_1}\dot{q}_2=\partial_{u_1}(r_3p_1-p_3r_1)=r_4p_1+r_3p_2-p_4r_1-p_3r_2.\]
From the equalities above, we have 
\begin{align*}
\wp_{(2g-1)\cdot3}&=\frac{1}{2}q_2(r_3p_1-p_3r_1)+(r_2p_1-p_2r_1)\left(4q_1\wp_{1\cdot2}+2\lambda_2q_1+p_1-\frac{1}{2}q_3\right)\\
&\quad -\frac{1}{2}q_1(r_4p_1+r_3p_2-p_4r_1-p_3r_2)+q_1(p_2q_1+p_1q_2). 
\end{align*}
From Theorem \ref{2024.11.20.1} (i) and (iii), we have 
\begin{align*}
\dot{p}_2&=q_3p_1-p_3q_1=6p_1^2-2p_1\wp_{3,2g-3}+2q_1\wp_{3,2g-1},\\
\dot{q}_2&=r_3p_1-p_3r_1=6p_1q_1-2p_1\wp_{3,2g-5}+2r_1\wp_{3,2g-1},\\
\dot{r}_2&=s_3p_1-p_3s_1=6p_1r_1-2p_1\wp_{3,2g-7}+2s_1\wp_{3,2g-1}+2\delta_{g,4}\lambda_4p_1,\\
\dot{p}_3&=\partial_{u_1}\dot{p}_2=10p_1p_2+2\wp_{1\cdot2}(p_1q_2-p_2q_1)+2q_2\wp_{3,2g-1}-2p_2\wp_{3,2g-3},\\
\dot{q}_3&=\partial_{u_1}\dot{q}_2=4p_1q_2+6p_2q_1+2\wp_{1\cdot2}(p_1r_2-p_2r_1)+2r_2\wp_{3,2g-1}-2p_2\wp_{3,2g-5},\\
\dot{r}_3&=\partial_{u_1}\dot{r}_2=4p_1r_2+6p_2r_1+2\wp_{1\cdot2}(p_1s_2-p_2s_1)+2s_2\wp_{3,2g-1}-2p_2\wp_{3,2g-7}+2\delta_{g,4}\lambda_4p_2,\\
p_4&=4p_1\wp_{1\cdot3}+8p_2\wp_{1\cdot2}+4\lambda_2p_2,\\
r_4&=4r_1\wp_{1\cdot3}+8r_2\wp_{1\cdot2}+4\lambda_2r_2+4q_2,\\
\dot{p}_4&=\partial_{u_1}\dot{p}_3=8p_2^2+10p_1p_3+4\wp_{1\cdot3}(p_1q_2-p_2q_1)+2\wp_{1\cdot2}(p_1q_3-p_3q_1)+2q_3\wp_{3,2g-1}\\
&\hspace{10ex} -2p_3\wp_{3,2g-3},\\
\dot{r}_4&=\partial_{u_1}\dot{r}_3=8p_2r_2+4p_1r_3+6p_3r_1+4\wp_{1\cdot3}(p_1s_2-p_2s_1)+2\wp_{1\cdot2}(p_1s_3-p_3s_1)\\
&\hspace{10ex} +2s_3\wp_{3,2g-1}-2p_3\wp_{3,2g-7}+2\delta_{g,4}\lambda_4p_3.
\end{align*}
From the equalities above, Theorem \ref{2024.11.20.1} (i), and (ii), we have $\wp_{(2g-1)\cdot4}=J\bigl(\{\wp_{i,j}\}\bigr)$, where $J\bigl(\{\wp_{i,j}\}\bigr)$ is a polynomial in $\{\wp_{i,j}\}$. 
Let 
\begin{align*}
Z_1&=(\wp_{1\cdot2,2g-5}\wp_{1\cdot2,2g-3})(\wp_{1\cdot2,2g-3}\wp_{1\cdot2,2g-1})-(\wp_{1\cdot2,2g-5}\wp_{1\cdot2,2g-1})(\wp_{1\cdot2,2g-3}\wp_{1\cdot2,2g-3}),\\
Z_2&=(\wp_{1\cdot2,2g-5}\wp_{1\cdot2,2g-1})(\wp_{1\cdot2,2g-5}\wp_{1\cdot2,2g-1})-(\wp_{1\cdot2,2g-5}\wp_{1\cdot2,2g-5})(\wp_{1\cdot2,2g-1}\wp_{1\cdot2,2g-1}),\\
Z_3&=(\wp_{1\cdot2,2g-7}\wp_{1\cdot2,2g-3})(\wp_{1\cdot2,2g-1}\wp_{1\cdot2,2g-1})-(\wp_{1\cdot2,2g-7}\wp_{1\cdot2,2g-1})(\wp_{1\cdot2,2g-3}\wp_{1\cdot2,2g-1}),
\end{align*}
where the parentheses mean that the substitutions from Theorem \ref{2024.11.20.1} (ii) are made before expanding (cf. \cite[Corollary 3.2.2]{BEL-97-1}). 
We have $Z_i=0$ for $i=1,2,3$. 
From the direct calculations, we can check that $J\bigl(\{\wp_{i,j}\}\bigr)+(Z_1+Z_2+Z_3)/2$ is equal to the right-hand side of the equality in Theorem \ref{2024.10.30.1}. 
\end{proof}

We consider the case $g\ge3$, assume $\lambda_{4g+2}\neq0$, and for $g\ge4$ take constants $b_i\in\mathbb{C}$ with $1\le i\le g-3$. 
Let 
\[\varphi(t_1,t_2,t_3)=-2\wp_{(2g-1)\cdot2}\left(b_1,\dots,b_{g-3}, \mathfrak{c}t_3, \mathfrak{d}t_2, t_1+\mathfrak{e}t_2\right)-\mathfrak{f},\]
where 
\[\mathfrak{c}=-16\lambda_{4g+2},\quad \mathfrak{d}=2\sqrt{-3\lambda_{4g+2}},\quad \mathfrak{e}=\frac{\lambda_{4g}}{\sqrt{-3\lambda_{4g+2}}},\quad \mathfrak{f}=\frac{2}{3}\lambda_{4g-2}+\frac{\lambda_{4g}^2}{18\lambda_{4g+2}}.\]

\begin{cor}\label{2024.12.6.1}
If $\lambda_{4g+2}\neq0$, the function $\varphi$ satisfies the KP-$\mathrm{I}$ equation
\[\partial_{t_1}(\partial_{t_3}\varphi+6\varphi\partial_{t_1}\varphi+\partial_{t_1}^3\varphi)=\partial_{t_2}^2\varphi.\]
\end{cor}

\begin{proof}
In \cite[p.~170]{BEL-2000}, it was pointed out that if $g\ge3$, under certain restrictions on the coefficients of the defining equation of the curve, $\wp_{(2g-1)\cdot2}$ is a solution to the KP equation. 
We give a proof of this corollary. 
By differentiating the equality in Theorem \ref{2024.10.30.1} with respect to $u_{2g-1}$ twice, we obtain
\begin{align*}
\wp_{(2g-1)\cdot6}&=12\wp_{(2g-1)\cdot3}^2+12\wp_{(2g-1)\cdot2}\wp_{(2g-1)\cdot4}+4\lambda_{4g}\wp_{(2g-1)\cdot3,2g-3}\\
&\quad +16\lambda_{4g+2}\wp_{(2g-1)\cdot3,2g-5}-12\lambda_{4g+2}\wp_{(2g-1)\cdot2, (2g-3)\cdot2}+4\lambda_{4g-2}\wp_{(2g-1)\cdot4}. 
\end{align*} 
We have 
\begin{gather*}
\partial_{t_1}\partial_{t_3}\varphi=-2\mathfrak{c}\wp_{(2g-1)\cdot3,2g-5},\qquad (\partial_{t_1}\varphi)^2=4\wp_{(2g-1)\cdot3}^2,\\
\varphi\partial_{t_1}^2\varphi=2(2\wp_{(2g-1)\cdot2}+\mathfrak{f})\wp_{(2g-1)\cdot4},\qquad \partial_{t_1}^4\varphi=-2\wp_{(2g-1)\cdot6},\\
\partial_{t_2}^2\varphi=-2\mathfrak{e}^2\wp_{(2g-1)\cdot4}-4\mathfrak{d}\mathfrak{e}\wp_{(2g-1)\cdot3, 2g-3}-2\mathfrak{d}^2\wp_{(2g-1)\cdot2, (2g-3)\cdot2}.
\end{gather*}
From the equalities above, we obtain the statement of the corollary. 
\end{proof}

\begin{rem}\label{2025.1.30.1}
In general, let $\Phi(t_1,t_2,t_3)$ be a solution to the KP-I equation
\[\partial_{t_1}(\partial_{t_3}\Phi+6\Phi\partial_{t_1}\Phi+\partial_{t_1}^3\Phi)=\partial_{t_2}^2\Phi.\]
Then the function $\widetilde{\Phi}(t_1,t_2,t_3)=\Phi(t_1,\sqrt{-1}t_2,t_3)$ is a solution to the KP-II equation 
\begin{equation}
\partial_{t_1}(\partial_{t_3}\widetilde{\Phi}+6\widetilde{\Phi}\partial_{t_1}\widetilde{\Phi}+\partial_{t_1}^3\widetilde{\Phi})=-\partial_{t_2}^2\widetilde{\Phi}.\label{2025.2.14.00876530298sdf}
\end{equation}
Let $\xi_8$ be a primitive 8th root of unity. 
The function $\overline{\Phi}(t_1,t_2,t_3)=\xi_8^2\Phi(\xi_8t_1,t_2,\xi_8^3t_3)$ is also a solution to the KP-II equation in the form (\ref{2025.2.14.00876530298sdf}).  
\end{rem}

\begin{rem}\label{2025.1.23.1234}
By the shift $X'=X+X_0$ for some $X_0\in\mathbb{C}$, it is always possible to achieve the condition $\lambda_{4g+2}\neq0$.  
Thus, we have obtained a solution to the KP equation for any hyperelliptic curve up to normalization. 
\end{rem}


\section{Hyperelliptic functions associated with a hyperelliptic curve with two infinite points}\label{2024.10.25.1}

In this section, we define basic meromorphic functions on the Jacobian variety of a hyperelliptic curve with two infinite points in accordance with \cite[p.~145]{B}. 

For a positive integer $g$, let us consider the polynomial in $x$ 
\[
N(x)=\nu_{0}x^{2g+2}+\nu_{2}x^{2g+1}+\cdots+\nu_{4g+2}x+\nu_{4g+4}, \quad \nu_i\in\mathbb{C},\quad\nu_{0}\neq0.
\]
We assume that $N(x)$ has no multiple roots and consider the non-singular hyperelliptic curve of genus $g$
\[V=\Bigl\{(x,y)\in\mathbb{C}^2 \Bigm| y^2=N(x)\Bigr\}.\]
We assign weights for $x$, $y$, and $\nu_i$ as $\wt(x)=2$, $\wt(y)=2g+2$, and $\wt(\nu_i)=i$.  
The equation $y^2=N(x)$ has the homogeneous weight $4g+4$ with respect to the coefficients $\nu_i$ and the variables $x,y$.
A basis of the vector space consisting of holomorphic 1-forms on $V$ is given by 
\begin{equation}
\mu_i=\frac{x^{i-1}}{2y}dx, \qquad 1\le i\le g. \label{2025.6.1.457637532final}
\end{equation}
We set $\mu={}^t(\mu_1,\dots,\mu_g)$. 
Let $\{\mathfrak{a}_i,\mathfrak{b}_i\}_{i=1}^g$ be a canonical basis in the one-dimensional homology group of the curve $V$. 
We define the period matrices by 
\[
2\mu'=\left(\int_{\mathfrak{a}_j}\mu_i\right),\qquad
2\mu''=\left(\int_{\mathfrak{b}_j}\mu_i\right). 
\]
We define the lattice of periods $L=\bigl\{2\mu'm_1+2\mu''m_2\mid m_1,m_2\in\mathbb{Z}^g\bigr\}$ and consider the Jacobian variety $\operatorname{Jac}(V)=\mathbb{C}^g/L$. 
We take $a\in\mathbb{C}$ such that $N(a)=0$. 
Let $\operatorname{Sym}^g(V)$ be the $g$-th symmetric product of $V$.   
Let $\mathcal{F}\bigl(\operatorname{Sym}^g(V)\bigr)$ and $\mathcal{F}\bigl(\operatorname{Jac}(V)\bigr)$ be the fields of meromorphic functions on $\operatorname{Sym}^g(V)$ and $\operatorname{Jac}(V)$, respectively. 
Let us consider the Abel-Jacobi map
\[I: \quad\operatorname{Sym}^g(V)\to\operatorname{Jac}(V),\qquad \sum_{i=1}^gQ_i\mapsto\sum_{i=1}^g\int_{(a,0)}^{Q_i}\mu.\]
The map $I$ induces the isomorphism of fields 
\[I^* :\quad \mathcal{F}\bigl(\operatorname{Jac}(V)\bigr)\to\mathcal{F}\bigl(\operatorname{Sym}^g(V)\bigr),\qquad \phi\mapsto \phi\circ I.\]
For $(x_i,y_i)\in V$ with $1\le i\le g$, let 
\[R(x)=(x-a)(x-x_1)\cdots(x-x_g),\qquad R'(x)=\frac{d}{dx}R(x).\]
For variables $e_1,e_2$, we set 
\begin{gather*}
\nabla=\sum_{i=1}^g\frac{y_i}{(e_1-x_i)(e_2-x_i)R'(x_i)},\quad f(e_1,e_2)=\sum_{i=0}^{g+1}e_1^ie_2^i\bigl\{2\nu_{4g+4-4i}+\nu_{4g+2-4i}(e_1+e_2)\bigr\},
\end{gather*}
where we set $\nu_{-2}=0$. 

\begin{lem}[{\cite[p.~146]{B}}]\label{2024.10.24.1}
We have 
\[f(e_1,e_1)=2N(e_1),\qquad \left.\frac{\partial f}{\partial e_2}\right|_{e_2=e_1}=\left.\frac{dN}{dx}\right|_{x=e_1}.\]
\end{lem}

\begin{lem}[{\cite[p.~315]{B2}}]\label{2024.11.27.1}
For a symmetric polynomial $\widehat{f}(e_1,e_2)\in\mathbb{C}[e_1,e_2]$, we assume that the degree of $\widehat{f}(e_1,e_2)$ is $g+1$ in each variable and 
\[\widehat{f}(e_1,e_1)=2N(e_1),\qquad \left.\frac{\partial \widehat{f}}{\partial e_2}\right|_{e_2=e_1}=\left.\frac{dN}{dx}\right|_{x=e_1}.\]
Then there exist complex numbers $\left\{\mathfrak{m}_{i,j}\right\}_{i,j=1}^g$ such that $\mathfrak{m}_{i,j}=\mathfrak{m}_{j,i}$ and
\[\widehat{f}(e_1,e_2)=f(e_1,e_2)+(e_1-e_2)^2\sum_{i,j=1}^g\mathfrak{m}_{i,j}e_1^{i-1}e_2^{j-1}.\]
\end{lem}

We set 
\[F(e_1,e_2)=f(e_1,e_2)R(e_1)R(e_2)+(e_1-e_2)^2R(e_1)^2R(e_2)^2\nabla^2-N(e_1)R(e_2)^2-N(e_2)R(e_1)^2.\]
Note that $F(e_1,e_2)$ is a symmetric polynomial in $e_1$ and $e_2$. 

\begin{lem}
The polynomial $F(e_1,e_2)$ can be divided by $R(e_1)R(e_2)$. 
\end{lem}

\begin{proof}
First, we will prove $F(x_1,e_2)=0$. 
We have  $f(x_1,e_2)R(x_1)R(e_2)=N(e_2)R(x_1)^2=0$ and $\nabla^2=F_1(e_1,e_2)+F_2(e_1,e_2)$, where 
\begin{align*}
F_1(e_1,e_2)&=\frac{N(x_1)}{(e_1-x_1)^2(e_2-x_1)^2R'(x_1)^2},\\
F_2(e_1,e_2)&=\sum_{i=2}^g\frac{N(x_i)}{(e_1-x_i)^2(e_2-x_i)^2R'(x_i)^2}\\
&\quad +\sum_{1\le i<j\le g}\frac{2y_iy_j}{(e_1-x_i)(e_1-x_j)(e_2-x_i)(e_2-x_j)R'(x_i)R'(x_j)}.
\end{align*}
We have $R(x_1)^2R(e_2)^2F_2(x_1,e_2)=0$ and 
\[R(e_1)^2R(e_2)^2F_1(e_1,e_2)=\frac{N(x_1)(e_1-a)^2(e_2-a)^2\prod_{i=2}^g(e_1-x_i)^2(e_2-x_i)^2}{R'(x_1)^2}.\]
Thus, we have 
\[(x_1-e_2)^2R(x_1)^2R(e_2)^2F_1(x_1,e_2)=N(x_1)(e_2-a)^2\prod_{i=1}^g(e_2-x_i)^2=N(x_1)R(e_2)^2.\]
Therefore, we have $F(x_1,e_2)=0$. Similarly, we have $F(x_i,e_2)=0$ for any $1\le i\le g$. 
From $N(a)=R(a)=0$, we can check $F(a,e_2)=0$. Therefore, the polynomial $F(e_1,e_2)$ can be divided by $R(e_1)$. 
Since $F(e_1,e_2)$ is a symmetric polynomial in $e_1$ and $e_2$, it can be divided by $R(e_2)$. Thus, we obtain the statement of the lemma.  
\end{proof}


\begin{lem}
The polynomial $F(e_1,e_2)$ can be divided by $(e_1-e_2)^2$. 
\end{lem}

\begin{proof}
From Lemma \ref{2024.10.24.1} and the direct calculations, we can check 
\[F(e_1,e_1)=0,\qquad \left.\frac{\partial F}{\partial e_2}\right|_{e_2=e_1}=0.\]
Thus, we obtain the statement of the lemma.  
\end{proof}

Let $G(e_1,e_2)=F(e_1,e_2)/\bigl\{(e_1-e_2)^2R(e_1)R(e_2)\bigr\}$. 
Then $G(e_1,e_2)$ is a symmetric polynomial in $e_1$ and $e_2$ of degree at most $g-1$ in each variable. 
We assign weights for $a$, $x_i$, $y_i$, and $e_i$ as $\wt(a)=\wt(x_i)=\wt(e_i)=2$ and $\wt(y_i)=2g+2$.   
Then $G(e_1,e_2)$ has the homogeneous weight $4g$.  

\begin{definition}[{\cite[p.~145]{B}}]\label{2024.11.29.12345}

\noindent (i) For $1\le i,j\le g$, we define $P_{2g+2-2i, 2g+2-2j}\in\mathcal{F}\bigl(\operatorname{Sym}^g(V)\bigr)$ by 
\begin{equation}
\sum_{i,j=1}^gP_{2g+2-2i, 2g+2-2j}e_1^{i-1}e_2^{j-1}=G(e_1,e_2).\label{2024.10.30.1234}
\end{equation}

\noindent (ii) For $i,j=2,4,\dots,2g$, we define the meromorphic functions $\mathcal{P}_{i,j}(v)$ with $v=(v_{2g},v_{2g-2},\dots,v_2)\in\mathbb{C}^g$ on $\operatorname{Jac}(V)$ by $\mathcal{P}_{i,j}=(I^*)^{-1}(P_{i,j})$. 
\end{definition}

\begin{rem}
Since $G(e_1,e_2)$ is a symmetric polynomial in $e_1$ and $e_2$, we have $\mathcal{P}_{i,j}=\mathcal{P}_{j,i}$ for any $i,j$. 
\end{rem}

\begin{ex}
For $g=1$, we have 
\[P_{2,2}=\frac{a(\nu_2+2a\nu_0)x_1+\nu_6+2a\nu_4+2a^2\nu_2+2a^3\nu_0}{x_1-a}.\]
\end{ex}

Let $\mathcal{P}_{i,j,k_1,\dots,k_l}=\partial_{v_{k_1}}\cdots\;\partial_{v_{k_l}}\mathcal{P}_{i,j}$.  
From (\ref{2024.10.30.1234}), we have $\wt(P_{i,j})=i+j$ with respect to $x_i$, $y_i$, $\nu_i$, and $a$.    
We assign weights for $\mathcal{P}_{i,j}$ as $\wt(\mathcal{P}_{i,j})=i+j$.    
For $1\le i\le g$, let 
\[\chi_i(x)=x^i-h_1(x_1,\dots,x_g)x^{i-1}+h_2(x_1,\dots,x_g)x^{i-2}-\cdots+(-1)^ih_i(x_1,\dots,x_g).\]
We have $\chi_g(x)=(x-x_1)\cdots(x-x_g)$. We set $\chi_0(x)=1$. 
Note that $\chi_i(x)$ has the homogeneous weight $2i$ with respect to $x$ and $x_i$.  
For $1\le j\le g$, we have 
\[\frac{\chi_g(x)}{(x-x_j)}=x^{g-1}+\chi_1(x_j)x^{g-2}+\chi_2(x_j)x^{g-3}+\cdots+\chi_{g-1}(x_j)\]
(cf. \cite[p.~136]{B}). 

\begin{lem}[{\cite[p.~137]{B}}]\label{2024.10.31.1}
By the isomorphism $I^*$, for $1\le i\le g$, the derivation $\partial_{v_{2g+2-2i}}$ of $\mathcal{F}\bigl(\operatorname{Jac}(V)\bigr)$ corresponds to the following derivation of $\mathcal{F}\bigl(\operatorname{Sym}^g(V)\bigr)$$:$
\[\sum_{j=1}^g\frac{2y_j}{\chi_g'(x_j)}\chi_{g-i}(x_j)\partial_{x_j}.\]
\end{lem}

From Lemma \ref{2024.10.31.1}, we can assign weights for $\mathcal{P}_{i_1,\dots,i_k}$ as $\wt(\mathcal{P}_{i_1,\dots,i_k})=i_1+\cdots+i_k$.

\section{Solution to the KP equation in terms of the hyperelliptic function associated with the curve $V$}

Let 
\[
E(e_1,e_2)=(e_1-e_2)\left\{f(e_1,e_2)-(e_1-e_2)^2\sum_{i,j=1}^g\mathcal{P}_{2g+2-2i, 2g+2-2j}e_1^{i-1}e_2^{j-1}\right\}.
\]

\begin{lem}[{\cite[p.~144]{B}}]\label{2024.10.21.111}
For variables $e_1,e_2,e_3,e_4$, the following equality holds$:$ 
\begin{gather*}
\frac{1}{2}(e_2-e_1)(e_3-e_2)(e_3-e_1)(e_4-e_3)(e_4-e_2)(e_4-e_1)\times\\
\sum_{i,j,k,l=1}^g\mathcal{P}_{2g+2-2i, 2g+2-2j, 2g+2-2k, 2g+2-2l}e_1^{i-1}e_2^{j-1}e_3^{k-1}e_4^{l-1}\\
=E(e_2,e_3)E(e_4,e_1)+E(e_3,e_1)E(e_4,e_2)+E(e_1,e_2)E(e_4,e_3).
\end{gather*}
\end{lem}

\begin{rem}
The equality in Lemma \ref{2024.10.21.111} has the homogeneous weight $8g+12$. 

\end{rem}

\begin{lem}[{\cite[p.~144]{B}}]
The functions $\mathcal{P}_{i,j,k,l}$ have values independent of the order of the suffixes $i,j,k,l$. 
\end{lem}

\begin{proof}
For the sake to be complete and self-contained, we give a proof of this lemma. Let 
\[\widetilde{E}(e_1,e_2,e_3,e_4)=E(e_2,e_3)E(e_4,e_1)+E(e_3,e_1)E(e_4,e_2)+E(e_1,e_2)E(e_4,e_3).\]
From $\widetilde{E}(e_1,e_1,e_3,e_4)=0$, the polynomial $\widetilde{E}(e_1,e_2,e_3,e_4)$ can be divided by $(e_2-e_1)$. 
Similarly, we find that the polynomial $\widetilde{E}(e_1,e_2,e_3,e_4)$ can be divided by $(e_2-e_1)(e_3-e_2)(e_3-e_1)(e_4-e_3)(e_4-e_2)(e_4-e_1)$. 
We can check that $\widetilde{E}(e_1,e_2,e_3,e_4)/\bigl\{(e_2-e_1)(e_3-e_2)(e_3-e_1)(e_4-e_3)(e_4-e_2)(e_4-e_1)\bigr\}$ is a symmetric polynomial in 
$e_1,e_2,e_3,e_4$. Therefore, we obtain the statement of the lemma.  
\end{proof}

\begin{rem}
In \cite[pp.~155, 156]{B}, for $g=1,2,3$, the functions $\mathcal{P}_{i,j,k,l}$ with $i,j,k,l=2,4,\dots,2g$ were expressed in terms of $\mathcal{P}_{m,n}$ explicitly. 
\end{rem}

\begin{rem}
In \cite[p.~156]{B}, for any $g$, the functions $\mathcal{P}_{i,j,k,l}$ with $i,j,k,l=2g-2,2g$ were expressed in terms of $\mathcal{P}_{m,n}$ explicitly. 
\end{rem}

If the number of $i$ is $d$ in the suffixes of $\mathcal{P}_{i_1,\dots,i_k}$, then we use the notation $i\cdot d$. 
For example, we denote $\mathcal{P}_{2,2}$, $\mathcal{P}_{2,2,4}$, and $\mathcal{P}_{2,2,4,4}$ by $\mathcal{P}_{2\cdot2}$, $\mathcal{P}_{2\cdot2,4}$, and $\mathcal{P}_{2\cdot2, 4\cdot2}$, respectively. 
If $i_j\notin\{2,4,\dots,2g\}$ for some $1\le j\le k$, we set $\mathcal{P}_{i_1,\dots,i_k}=0$. 

\begin{theorem}\label{2024.10.22.1111}
For $1\le k\le g$, we have the following equalities$:$
\begin{align*}
\mathcal{P}_{2\cdot3,2k}&=2\mathcal{P}_{2,2k}(3\mathcal{P}_{2\cdot2}+2\nu_{4})+2\nu_{2}(\delta_{1,k}\nu_{6}-\mathcal{P}_{4,2k}+3\mathcal{P}_{2, 2k+2})\\
&\quad +4\nu_{0}(3\mathcal{P}_{2, 2k+4}-3\mathcal{P}_{4,2k+2}+\mathcal{P}_{6,2k}-2\delta_{1,k}\nu_{8}-\delta_{2,k}\nu_{10}). 
\end{align*}
\end{theorem}

\begin{proof}
For $1\le i\le g$, the coefficient of $e_1^{i-1}e_2^{g}e_3^{g+1}e_4^{g+2}$ in the left-hand side of the equality in Lemma \ref{2024.10.21.111} is $\mathcal{P}_{2\cdot3,2g+2-2i}/2$. 
For $1\le i\le g$, the coefficients of $e_1^{i-1}e_2^{g}e_3^{g+1}e_4^{g+2}$ in $E(e_2,e_3)E(e_4,e_1)$, $E(e_3,e_1)E(e_4,e_2)$, and $E(e_1,e_2)E(e_4,e_3)$ are $(3\mathcal{P}_{2\cdot2}+2\nu_{4})\mathcal{P}_{2,2g+2-2i}$, 
$\nu_{2}(\delta_{g,i}\nu_{6}-\mathcal{P}_{4, 2g+2-2i}+3\mathcal{P}_{2, 2g+4-2i})$, and 
$2\nu_0(3\mathcal{P}_{2, 2g+6-2i}-3\mathcal{P}_{4,2g+4-2i}+\mathcal{P}_{6,2g+2-2i}-2\delta_{g,i}\nu_{8}-\delta_{g-1,i}\nu_{10})$, respectively. From Lemma \ref{2024.10.21.111}, we obtain the statement of the theorem. 
\end{proof}


\begin{rem}
The equality in Theorem \ref{2024.10.22.1111} has the homogeneous weight $2k+6$. 

\end{rem}

We consider the case $g\ge3$ and for $g\ge4$ take constants $c_i\in\mathbb{C}$ with $1\le i\le g-3$. 
Let 
\[\psi(t_1,t_2,t_3)=-2\mathcal{P}_{2\cdot2}\left(c_1,\dots,c_{g-3},\alpha t_3, \beta t_2, t_1+\gamma t_2\right)-\delta,\]
where 
\[\alpha=-16\nu_{0},\quad \beta=2\sqrt{-3\nu_{0}},\quad \gamma=\frac{\nu_{2}}{\sqrt{-3\nu_{0}}},\quad \delta=\frac{2}{3}\nu_{4}+\frac{\nu_{2}^2}{18\nu_{0}}.\]

\begin{cor}\label{4}
The function $\psi$ satisfies the KP-$\mathrm{I}$ equation
\[\partial_{t_1}(\partial_{t_3}\psi+6\psi\partial_{t_1}\psi+\partial_{t_1}^3\psi)=\partial_{t_2}^2\psi.\]
\end{cor}

\begin{proof}
By differentiating the equality in Theorem \ref{2024.10.22.1111} with $k=1$ with respect to $v_2$ twice, we obtain
\[\mathcal{P}_{2\cdot6}=12\mathcal{P}_{2\cdot3}^2+12\mathcal{P}_{2\cdot2}\mathcal{P}_{2\cdot4}+4\nu_{4}\mathcal{P}_{2\cdot4}+4\nu_{2}\mathcal{P}_{2\cdot3,4}+16\nu_{0}\mathcal{P}_{2\cdot3,6}-12\nu_{0}\mathcal{P}_{2\cdot2,4\cdot2}.\]
We have 
\begin{gather*}
\partial_{t_1}\partial_{t_3}\psi=-2\alpha\mathcal{P}_{2\cdot3,6},\qquad (\partial_{t_1}\psi)^2=4\mathcal{P}_{2\cdot3}^2,\qquad \psi\partial_{t_1}^2\psi=2(2\mathcal{P}_{2\cdot2}+\delta)\mathcal{P}_{2\cdot4},\\
\partial_{t_1}^4\psi=-2\mathcal{P}_{2\cdot6},\qquad \partial_{t_2}^2\psi=-2\beta^2\mathcal{P}_{2\cdot2,4\cdot2}-4\beta\gamma\mathcal{P}_{2\cdot3,4}-2\gamma^2\mathcal{P}_{2\cdot4}.
\end{gather*}
From the equalities above, we obtain the statement of the corollary. 
\end{proof}



\begin{ex}
Let us consider the case $N(x)=x^{2g+2}+\nu_{4g+4}$. 
We have $f(e_1,e_2)=2(e_1^{g+1}e_2^{g+1}+\nu_{4g+4})$. 
From Theorem \ref{2024.10.22.1111}, for $g=1$, we have 
\[\mathcal{P}_{2\cdot4}=6\mathcal{P}_{2\cdot2}^2-8\nu_8.\]
For $g\ge2$ and $1\le k\le g$, we have 
\[\mathcal{P}_{2\cdot3,2k}=6\mathcal{P}_{2\cdot2}\mathcal{P}_{2,2k}+4(3\mathcal{P}_{2, 2k+4}-3\mathcal{P}_{4,2k+2}+\mathcal{P}_{6,2k}).\]
From Corollary \ref{4}, the function 
\[\psi(t_1,t_2,t_3)=-2\mathcal{P}_{2\cdot2}\left(c_1,\dots,c_{g-3},-16t_3, 2\sqrt{-3} t_2, t_1\right)\]
satisfies the KP-I equation
\[\partial_{t_1}(\partial_{t_3}\psi+6\psi\partial_{t_1}\psi+\partial_{t_1}^3\psi)=\partial_{t_2}^2\psi.\]
\end{ex}

\section{Inversion problem of the Abel-Jacobi map for the curve $V$}\label{2024.11.4.1}

We consider the curve $V$ defined in Section \ref{2024.10.25.1} and set 
\[N(x)=\nu_{0}(x-a)\prod_{i=1}^{2g+1}(x-a_i),\qquad a_i\in\mathbb{C}.\]
We take $\mathfrak{s}, \mathfrak{t}\in\mathbb{C}$ such that $\mathfrak{s}\mathfrak{t}\neq0$ and $\mathfrak{s}^{2g+1}/\mathfrak{t}^2=N'(a)$. 
We assign weights for $\mathfrak{s}$ and $\mathfrak{t}$ as $\wt(\mathfrak{s})=4$ and $\wt(\mathfrak{t})=2g+1$. 
Let us consider the polynomial 
\begin{equation}
\widetilde{M}(X)=\prod_{i=1}^{2g+1}\left(X-\frac{\mathfrak{s}}{a_i-a}\right)\label{2025.6.5.5387549100}
\end{equation}
and the hyperelliptic curve $\widetilde{C}$ of genus $g$ defined by 
\[\widetilde{C}=\Bigl\{(X,Y)\in\mathbb{C}^2 \Bigm| Y^2=\widetilde{M}(X)\Bigr\}.\]

\begin{prop}[{\cite[pp.~144, 145]{B}}]\label{3}
We have the following isomorphism from $V$ to $\widetilde{C}$$:$
\begin{equation}
\zeta\colon\quad V\to \widetilde{C},\qquad (x,y)\mapsto (X,Y)=\left(\frac{\mathfrak{s}}{x-a}, \frac{\mathfrak{t}\;y}{(x-a)^{g+1}}\right).\label{1}
\end{equation}
\end{prop}



Let $D$ be the $g\times g$ regular matrix defined by ${}^t\bigl(\zeta^*(\omega_1),\dots,\zeta^*(\omega_g)\bigr)=D\mu$, where $\zeta^*(\omega_i)$ is the pullback of the holomorphic 1-form $\omega_i$ on $\widetilde{C}$ with respect to the map $\zeta$. 
For $1\le i\le g$, we take points $T_i=(x_i,y_i)\in V\backslash\bigl\{\pm\infty, (a,0)\bigr\}$ such that $T_i\neq\tau_2(T_j)$ if $i\neq j$, where $\tau_2$ is the hyperelliptic involution of $V$
\[\tau_2: V\to V,\qquad (x,y)\mapsto(x,-y).\] 
Let 
\[v=\sum_{i=1}^g\int_{(a,0)}^{T_i}\mu.\]

\begin{prop}\label{55555}
For $1\le k\le g$, we have 
\[h_k(x_1-a,\dots,x_g-a)=(-\mathfrak{s})^k\frac{\wp_{1,2g-2k-1}(Dv)}{\wp_{1,2g-1}(Dv)},\]
where $\wp_{1,2i-1}$ with $1\le i\le g$ are the hyperelliptic functions associated with the curve $\widetilde{C}$ and we set $\wp_{1,-1}(Dv)=-1$. 
\end{prop}

\begin{proof}
For $1\le i\le g$, let $(X_i,Y_i)=\zeta\bigl((x_i,y_i)\bigr)$. We have 
\[v=\sum_{i=1}^g\int_{\infty}^{(X_i,Y_i)}D^{-1}\omega=D^{-1}\sum_{i=1}^g\int_{\infty}^{(X_i,Y_i)}\omega.\]
Thus, we have 
\[\sum_{i=1}^g\int_{\infty}^{(X_i,Y_i)}\omega=Dv.\]
From Theorem \ref{2024.10.22.1}, we have 
\[h_k(X_1,\dots,X_g)=(-1)^{k-1}\wp_{1, 2k-1}(Dv),\qquad 0\le k\le g.\]
From Proposition \ref{3}, for $1\le k\le g$, we have 
\begin{align*}
h_k(x_1-a,\dots,x_g-a)&=\mathfrak{s}^k\frac{h_{g-k}(X_1,\dots,X_g)}{h_g(X_1,\dots,X_g)}\\
&=(-\mathfrak{s})^k\frac{\wp_{1, 2g-2k-1}(Dv)}{\wp_{1, 2g-1}(Dv)}.
\end{align*}

\end{proof}

\begin{rem}
Proposition $\ref{55555}$ is a generalization of \cite[Sections 20.6 and 20.7]{WW} for $g=1$ to any $g$. 
\end{rem}

\section{Relationships between the hyperelliptic functions associated with the curves $V$ and $\widetilde{C}$}\label{2024.12.7.1123}

We consider the hyperelliptic functions $\wp_{i,j}$ associated with the curve $\widetilde{C}$. 
For $0\le i\le 2g+1$, we define $\widetilde{\lambda}_{2i}\in\mathbb{C}$ such that the following equality holds:
\[\widetilde{M}(X)=\widetilde{\lambda}_0X^{2g+1}+\widetilde{\lambda}_2X^{2g}+\widetilde{\lambda}_4X^{2g-1}+\cdots+\widetilde{\lambda}_{4g}X+\widetilde{\lambda}_{4g+2}.\]
Let 
\[\widetilde{R}(X)=(X-X_1)\cdots(X-X_g),\]
where $X_1,\dots, X_g$ are defined in Section \ref{2024.11.4.1}. For variables $\widetilde{e}_1,\widetilde{e}_2$, we set 
\begin{gather*}
\widetilde{\nabla}=\sum_{i=1}^g\frac{Y_i}{(\widetilde{e}_1-X_i)(\widetilde{e}_2-X_i)\widetilde{R}'(X_i)},\\
\widetilde{f}(\widetilde{e}_1,\widetilde{e}_2)=\sum_{i=0}^{g}(\widetilde{e}_1)^i(\widetilde{e}_2)^i\left\{2\widetilde{\lambda}_{4g+2-4i}+\widetilde{\lambda}_{4g-4i}(\widetilde{e}_1+\widetilde{e}_2)\right\}.
\end{gather*}

\begin{lem}[{\cite[p.~146]{B}}]\label{2024.11.28.1}
We have 
\[\widetilde{f}(\widetilde{e}_1,\widetilde{e}_1)=2\widetilde{M}(\widetilde{e}_1),\qquad \left.\frac{\partial \widetilde{f}}{\partial \widetilde{e}_2}\right|_{\widetilde{e}_2=\widetilde{e}_1}=\left.\frac{d\widetilde{M}}{dX}\right|_{X=\widetilde{e}_1}.\]
\end{lem}

\begin{theorem}[{\cite[p.~138]{B}}, {\cite[pp.~328, 329]{B2}}, {\cite[p.~4729]{M}}]\label{2024.11.5.1}
The following equality holds$:$
\begin{gather*}
\sum_{i,j=1}^g\wp_{2g+1-2i, 2g+1-2j}(Dv)(\widetilde{e}_1)^{i-1}(\widetilde{e}_2)^{j-1}\\
=\frac{\widetilde{f}(\widetilde{e}_1,\widetilde{e}_2)}{(\widetilde{e}_1-\widetilde{e}_2)^2}+\widetilde{R}(\widetilde{e}_1)\widetilde{R}(\widetilde{e}_2)\widetilde{\nabla}^2
-\frac{1}{(\widetilde{e}_1-\widetilde{e}_2)^2}\left(\frac{\widetilde{M}(\widetilde{e}_1)\widetilde{R}(\widetilde{e}_2)}{\widetilde{R}(\widetilde{e}_1)}+\frac{\widetilde{M}(\widetilde{e}_2)\widetilde{R}(\widetilde{e}_1)}{\widetilde{R}(\widetilde{e}_2)}\right).
\end{gather*}
\end{theorem}

\begin{prop}\label{2024.11.30.9065376}
$(\mathrm{i})$ For $g^2$ variables $z=\{z_{2k-1,2l-1}\}_{k,l=1}^g$, there exist $g^2$ polynomials $K_{i,j}(z)\in\mathbb{C}[z]$ with $1\le i,j\le g$ such that the degree of $K_{i,j}(z)$ is $1$ in the variables $z$ and 
\[\mathcal{P}_{2g+2-2i,2g+2-2j}(v)=K_{i,j}\Bigl(\bigl\{\wp_{2k-1, 2l-1}(Dv)\bigr\}_{k,l=1}^g\Bigr),\qquad 1\le i,j\le g,\]
where in this equality we substitute $\wp_{2k-1,2l-1}(Dv)$ into $z_{2k-1,2l-1}$ for any $k$ and $l$. 

\vspace{1ex}

\noindent$(\mathrm{ii})$ If $a=0$, then the $(i,j)$ element of $D$ is $(\mathfrak{s}^{g+1-i}/\mathfrak{t})\delta_{i,j}$ for $1\le i,j\le g$ and we have 
\[\mathcal{P}_{2g+2-2i,2g+2-2j}(v)=\mathfrak{s}^{2g-i-j+2}\mathfrak{t}^{-2}\wp_{2i-1,2j-1}(Dv),\qquad 1\le i,j\le g.\]

\vspace{1ex}

\noindent$(\mathrm{iii})$ If $a=0$, $\nu_{4g+2}=1$, and $\mathfrak{s}=1$, then we have 
\[\mathcal{P}_{2g+2-2i,2g+2-2j}(v)=\wp_{2i-1,2j-1}(v),\qquad 1\le i,j\le g.\]
\end{prop}

\begin{proof}
We substitute $\widetilde{e}_i=\mathfrak{s}/(e_i-a)$ for $i=1,2$ into the both sides of the equality in Theorem \ref{2024.11.5.1} and multiply this equality by $(e_1-a)^{g-1}(e_2-a)^{g-1}$. 
We have 
\begin{equation}
(e_1-a)^{g-1}(e_2-a)^{g-1}\frac{\widetilde{f}(\widetilde{e}_1,\widetilde{e}_2)}{(\widetilde{e}_1-\widetilde{e}_2)^2}=\mathfrak{s}^{-2}\mathfrak{t}^2\frac{\overline{f}(e_1,e_2)}{(e_1-e_2)^2},\label{2024.11.30.1}
\end{equation}
where 
\[\overline{f}(e_1,e_2)=\mathfrak{t}^{-2}(e_1-a)^{g+1}(e_2-a)^{g+1}\widetilde{f}\left(\frac{\mathfrak{s}}{e_1-a}, \frac{\mathfrak{s}}{e_2-a}\right).\]
The polynomial $\overline{f}(e_1,e_2)$ is a symmetric polynomial in $e_1$ and $e_2$. 
From $\widetilde{\lambda}_{4g+2}\neq0$, the degree of $\overline{f}(e_1,e_2)$ is $g+1$ in each variable. 
From (\ref{2025.6.5.5387549100}), we have 
\begin{equation}
(e_1-a)^{2g+2}\widetilde{M}\left(\frac{\mathfrak{s}}{e_1-a}\right)=\mathfrak{t}^2N(e_1).\label{2024.11.28.1594}
\end{equation}
From Lemma \ref{2024.11.28.1} and (\ref{2024.11.28.1594}), we have $\overline{f}(e_1,e_1)=2N(e_1)$. 
By differentiating the both sides of (\ref{2024.11.28.1594}) with respect to $e_1$, we have 
\begin{equation}
(2g+2)(e_1-a)^{2g+1}\widetilde{M}\left(\frac{\mathfrak{s}}{e_1-a}\right)-\mathfrak{s}(e_1-a)^{2g}\widetilde{M}'\left(\frac{\mathfrak{s}}{e_1-a}\right)=\mathfrak{t}^2N'(e_1).\label{2024.11.28.209}
\end{equation}
We have 
\begin{align*}
\frac{\partial\overline{f}(e_1,e_2)}{\partial e_2}&=(g+1)\mathfrak{t}^{-2}(e_1-a)^{g+1}(e_2-a)^{g}\widetilde{f}\left(\frac{\mathfrak{s}}{e_1-a}, \frac{\mathfrak{s}}{e_2-a}\right)\\
&\quad -\mathfrak{s}\mathfrak{t}^{-2}(e_1-a)^{g+1}(e_2-a)^{g-1}\frac{\partial \widetilde{f}}{\partial \widetilde{e}_2}\left(\frac{\mathfrak{s}}{e_1-a}, \frac{\mathfrak{s}}{e_2-a}\right). 
\end{align*}
From Lemma \ref{2024.11.28.1} and (\ref{2024.11.28.209}), we have 
\[\left.\frac{\partial \overline{f}}{\partial e_2}\right|_{e_2=e_1}=N'(e_1).\]
From Lemma \ref{2024.11.27.1}, there exist complex numbers $\left\{\mathfrak{n}_{i,j}\right\}_{i,j=1}^g$ such that $\mathfrak{n}_{i,j}=\mathfrak{n}_{j,i}$ and
\[\overline{f}(e_1,e_2)=f(e_1,e_2)+(e_1-e_2)^2\sum_{i,j=1}^g\mathfrak{n}_{i,j}e_1^{i-1}e_2^{j-1}.\]
From (\ref{1}), we can check
\begin{gather}
(e_1-a)^{g-1}(e_2-a)^{g-1}\widetilde{R}(\widetilde{e}_1)\widetilde{R}(\widetilde{e}_2)\widetilde{\nabla}^2=\mathfrak{s}^{-2}\mathfrak{t}^2R(e_1)R(e_2)\nabla^2,\label{2024.11.30.2}\\
(e_1-a)^{g-1}(e_2-a)^{g-1}\frac{\widetilde{M}(\widetilde{e}_1)\widetilde{R}(\widetilde{e}_2)}{(\widetilde{e}_1-\widetilde{e}_2)^2\widetilde{R}(\widetilde{e}_1)}=\mathfrak{s}^{-2}\mathfrak{t}^2\frac{N(e_1)R(e_2)}{(e_1-e_2)^2R(e_1)},\label{2024.11.30.3}\\
(e_1-a)^{g-1}(e_2-a)^{g-1}\frac{\widetilde{M}(\widetilde{e}_2)\widetilde{R}(\widetilde{e}_1)}{(\widetilde{e}_1-\widetilde{e}_2)^2\widetilde{R}(\widetilde{e}_2)}=\mathfrak{s}^{-2}\mathfrak{t}^2\frac{N(e_2)R(e_1)}{(e_1-e_2)^2R(e_2)}.\label{2024.11.30.4}
\end{gather}
From Theorem \ref{2024.11.5.1}, for $g^2$ variables $z=\{z_{2k-1,2l-1}\}_{k,l=1}^g$, there exist $g^2$ polynomials $K_{i,j}(z)\in\mathbb{C}[z]$ with $1\le i,j\le g$ such that the degree of $K_{i,j}(z)$ is $1$ in the variables $z$ and 
\[\sum_{i,j=1}^gK_{i,j}\Bigl(\bigl\{\wp_{2k-1, 2l-1}(Dv)\bigr\}_{k,l=1}^g\Bigr)e_1^{i-1}e_2^{j-1}=G(e_1,e_2).\]
From Definition \ref{2024.11.29.12345}, we obtain the statement of (i). 
Let us consider the case $a=0$. 
The $(i,j)$ element of $D$ is $(\mathfrak{s}^{g+1-i}/\mathfrak{t})\delta_{i,j}$ for $1\le i,j\le g$. 
We have 
\begin{gather}
\begin{split}
&e_1^{g-1}e_2^{g-1}\sum_{i,j=1}^g\wp_{2g+1-2i, 2g+1-2j}(Dv)(\widetilde{e}_1)^{i-1}(\widetilde{e}_2)^{j-1}\\
&=\sum_{i,j=1}^g\mathfrak{s}^{2g-i-j}\wp_{2i-1, 2j-1}(Dv)e_1^{i-1}e_2^{j-1}.\label{2024.11.30.10984}
\end{split}
\end{gather}
The polynomial $\overline{f}(e_1,e_2)$ has the form  
\[\overline{f}(e_1,e_2)=\sum_{i=0}^{g+1}e_1^ie_2^i\bigl\{2\rho_{4g+4-4i}+\rho_{4g+2-4i}(e_1+e_2)\bigr\},\qquad \rho_j\in\mathbb{C},\quad \rho_{-2}=0.\]
From $\overline{f}(e_1,e_1)=2N(e_1)$, we have $\rho_j=\nu_j$ for any $j$. 
Thus, we have $\overline{f}(e_1,e_2)=f(e_1,e_2)$. 
From Theorem \ref{2024.11.5.1}, (\ref{2024.11.30.1}), and (\ref{2024.11.30.2})--(\ref{2024.11.30.10984}), we have 
\[\sum_{i,j=1}^g\mathfrak{s}^{2g-i-j}\wp_{2i-1, 2j-1}(Dv)e_1^{i-1}e_2^{j-1}=\mathfrak{s}^{-2}\mathfrak{t}^2G(e_1,e_2).\]
From Definition \ref{2024.11.29.12345}, we obtain the statement of (ii). 
Let us consider the case $a=0$, $\nu_{4g+2}=1$, and $\mathfrak{s}=1$. 
We have $N'(0)=\nu_{4g+2}=1$. Thus, we have $\mathfrak{t}^2=1$. 
For $1\le i,j\le g$, the function $\wp_{2i-1,2j-1}$ is an even function. 
From (ii), we obtain the statement of (iii). 
\end{proof}

\begin{rem}
Since we arranged the holomorphic 1-forms on $V$ in the form (\ref{2025.6.1.457637532final}), if $a=0$, the matrix $D$ is a diagonal matrix. 
Thus, the formulas in Proposition \ref{2024.11.30.9065376} (ii) and (iii) are simple.
\end{rem}

\begin{cor}
We have 
\[\mathcal{P}_{2\cdot2}(v)=\mathfrak{s}^2\mathfrak{t}^{-2}\wp_{(2g-1)\cdot2}(Dv)-\kappa,\]
where $\kappa=\mathfrak{t}^{-2}\bigl\{a^2g(g+1)\widetilde{\lambda}_{4g+2}-ag\mathfrak{s}\widetilde{\lambda}_{4g}\bigr\}$. 
\end{cor}

\begin{proof}
Since the coefficients of $e_1^{g+1}e_2^{g-1}$ in $\overline{f}(e_1,e_2)$ and $f(e_1,e_2)$ are $\kappa$ and $0$, respectively, 
we have $\mathfrak{n}_{g,g}=\kappa$. From Proposition \ref{2024.11.30.9065376} (i), we obtain the statement of the corollary. 
\end{proof}

\begin{ex}
Let us consider the case $g=2$. 
We have 
\[D=\frac{\mathfrak{s}}{\mathfrak{t}}\begin{pmatrix}\mathfrak{s}&0\\-a&1\end{pmatrix}.\] 
From Proposition \ref{2024.11.30.9065376} (i), we have 
\begin{align*}
\mathcal{P}_{2,4}(v)&=\mathfrak{t}^{-2}\bigl\{\mathfrak{s}^3\wp_{1,3}(Dv)-a\mathfrak{s}^2\wp_{3\cdot2}(Dv)-a^2\mathfrak{s}\widetilde{\lambda}_8+2a^3\widetilde{\lambda}_{10}\bigr\},\\
\mathcal{P}_{4\cdot2}(v)&=\mathfrak{t}^{-2}\bigl\{\mathfrak{s}^4\wp_{1\cdot2}(Dv)-2a\mathfrak{s}^3\wp_{1,3}(Dv)+a^2\mathfrak{s}^2\wp_{3\cdot2}(Dv)\\
&\quad +a\mathfrak{s}^3\widetilde{\lambda}_4-2a^2\mathfrak{s}^2\widetilde{\lambda}_6+4a^3\mathfrak{s}\widetilde{\lambda}_8-6a^4\widetilde{\lambda}_{10}\bigr\}.
\end{align*}
\end{ex}

\section*{Acknowledgments}

The authors would like to thank the referees for reading our manuscript carefully and giving useful comments. 
The authors are grateful to Petr Grinevich for useful discussions of the results of this work, Shigeki Matsutani for useful discussions of the work of Baker, and Atsushi Nakayashiki for useful comments on the KP equation. 
The work of Takanori Ayano was supported by JSPS KAKENHI Grant Number JP21K03296 and was partly supported by MEXT Promotion of Distinctive Joint Research Center Program JPMXP0723833165.


\begin{thebibliography}{00}


\bibitem{A-E-E-2003}
C. Athorne, J. C. Eilbeck, V. Z. Enolskii, Identities for the classical genus two $\wp$ function, J. Geom. Phys. \textbf{48} (2003), 354--368. 

https://doi.org/10.1016/S0393-0440(03)00048-2

\bibitem{A-E-E-2004}
C. Athorne, J. C. Eilbeck, V. Z. Enolskii, A SL(2) covariant theory of genus 2 hyperelliptic functions, Math. Proc. Camb. Phil. Soc. \textbf{136} (2004), 269--286. 

https://doi.org/10.1017/S030500410300728X

\bibitem{Ayano-Nakayashiki-2013}
T. Ayano, A. Nakayashiki, On Addition Formulae for Sigma Functions of Telescopic Curves, SIGMA \textbf{9} (2013), 046, 14 pages. 
https://doi.org/10.3842/SIGMA.2013.046

\bibitem{B-1898}
H. F. Baker, On the Hyperelliptic Sigma Functions, Amer. J. Math. \textbf{20} (1898), 301--384. https://doi.org/10.2307/2369512

\bibitem{B}
H. F. Baker, On a system of differential equations leading to periodic functions, Acta Math. \textbf{27} (1903), 135--156. 
https://doi.org/10.1007/BF02421301

\bibitem{B-1907}
H. F. Baker, An introduction to the theory of multiply periodic functions, Cambridge University Press, Cambridge, 1907.

\bibitem{B2}
H. F. Baker, Abelian Functions. Abel's theorem and the allied theory of theta functions, Cambridge Mathematical Library, Cambridge University Press, 1995. 


\bibitem{BEL-97-1}    
V. M. Buchstaber, V. Z. Enolskiĭ, D. V. Leĭkin, Hyperelliptic Kleinian Functions and Applications, Solitons, Geometry, and Topology: On the Crossroad, Amer. Math. Soc. Transl. Ser. 2, 179, Amer. Math. Soc., Providence, RI, 1997, 1--33. 

\bibitem{BEL-97-2}
V. M. Buchstaber, V. Z. Enolskii, D. V. Leykin, Kleinian Functions, Hyperelliptic Jacobians and Applications, Rev. Math. Math. Phys. 10, 1997, 3--120. 

\bibitem{BEL-2000}
V. M. Buchstaber, V. Z. Enolskii, D. V. Leykin, Uniformization of Jacobi Varieties of Trigonal Curves and Nonlinear Differential Equations, Funct. Anal. Appl. \textbf{34} (2000), 159--171. 
https://doi.org/10.1007/BF02482405

\bibitem{BEL-2012}  
V. M. Buchstaber, V. Z. Enolski, D. V. Leykin, Multi-Dimensional Sigma-Functions, arXiv:1208.0990, (2012). 
https://doi.org/10.48550/arXiv.1208.0990


\bibitem{BEL-2018}
V. M. Buchstaber, V. Z. Enolski, D. V. Leykin, $\sigma$-Functions: Old and New Results, Integrable Systems and Algebraic Geometry, 2, London Math. Soc. Lecture Note Ser. 459, Cambridge University Press, 2020, 175--214. 

https://doi.org/10.1017/9781108773355.007

\bibitem{BL-2004-Heat-Equations}
V. M. Buchstaber, D. V. Leykin, Heat Equations in a Nonholonomic Frame, Funct. Anal. Appl. \textbf{38} (2004), 88--101. 

https://doi.org/10.1023/B:FAIA.0000034039.92913.8a



\bibitem{Buchstaber-Mikhailov-2021}
V. M. Buchstaber, A. V. Mikhailov, Integrable polynomial Hamiltonian systems and symmetric powers of plane algebraic curves, Russ. Math. Surv. \textbf{76} (2021), 587--652. https://doi.org/10.1070/RM10007






\bibitem{Harnad-Enolski-2011}
J. Harnad, V. Z. Enolski, Schur function expansions of KP $\tau$-functions associated to algebraic curves, Russ. Math. Surv. \textbf{66} (2011), 767--807. 

https://doi.org/10.1070/RM2011v066n04ABEH004755

\bibitem{HST}
M. Hayashi, K. Shigemoto, T. Tsukioka, Differential equations of genus four hyperelliptic $\wp$ functions, J. Phys. Commun. \textbf{5} (2021), 105008. 

https://doi.org/10.1088/2399-6528/ac2bcf

\bibitem{Its-Matveev-1975}
A. R. Its, V. B. Matveev, Schr\"{o}dinger operators with finite-gap spectrum and N-soliton solutions of the Korteweg-de Vries equation, Theor. Math. Phys. \textbf{23} (1975), 343--355. https://doi.org/10.1007/BF01038218

\bibitem{Kadomtsev-Petviashvili-1970}
B. B. Kadomtsev, V. I. Petviashvili, On the stability of solitary waves in weakly dispersive media, Sov. Phys. Dokl. \textbf{15} (1970), 539--541. 

\bibitem{Khan-Akbar-2015}
K. Khan, M. A. Akbar, Exact traveling wave solutions of Kadomtsev-Petviashvili equation, J. Egypt. Math. Soc. \textbf{23} (2015), 278--281. 

https://doi.org/10.1016/j.joems.2014.03.010

\bibitem{Kl1}
F. Klein, Ueber hyperelliptische Sigmafunctionen, Math. Ann. \textbf{27} (1886), 431--464. 

https://doi.org/10.1007/BF01445285

\bibitem{Kl2}
F. Klein, Ueber hyperelliptische Sigmafunctionen, Math. Ann. \textbf{32} (1888), 351--380. 

https://doi.org/10.1007/BF01443606

\bibitem{Kodama-2018}
Y. Kodama, Solitons in Two-Dimensional Shallow Water, CBMS-NSF Regional Conference Series in Applied Mathematics, 92, Society for Industrial and Applied Mathematics, Philadelphia, 2018. 
https://doi.org/10.1137/1.9781611975529

\bibitem{Kodama-Xie-2021}
Y. Kodama, Y. Xie, Space Curves and Solitons of the KP Hierarchy. I. The $l$-th Generalized KdV Hierarchy, SIGMA \textbf{17} (2021), 024, 43 pages. 

https://doi.org/10.3842/SIGMA.2021.024


\bibitem{Krichever-1977}
I. M. Krichever, Methods of algebraic geometry in the theory of non-linear equations, Russ. Math. Surv. \textbf{32} (1977), 185--213. 

https://doi.org/10.1070/RM1977v032n06ABEH003862

\bibitem{Krichever-2008}
I. M. Krichever, Abelian solutions of the soliton equations and Riemann-Schottky problems, Russ. Math. Surv. \textbf{63} (2008), 1011--1022. 

https://doi.org/10.1070/RM2008v063n06ABEH004576

\bibitem{Kumar-Tiwari-Kumar-2017}
M. Kumar, A. K. Tiwari, R. Kumar, Some more solutions of Kadomtsev–Petviashvili equation, Comput. Math. Appl. \textbf{74} (2017), 2599--2607. 

https://doi.org/10.1016/j.camwa.2017.07.034

\bibitem{C-Liu-2010}
C. Liu, New exact periodic solitary wave solutions for Kadomtsev–Petviashvili equation, Appl. Math. Comput. \textbf{217} (2010), 1350--1354. 

https://doi.org/10.1016/j.amc.2009.04.080

\bibitem{M}
S. Matsutani, Hyperelliptic solutions of KdV and KP equations: re-evaluation of Baker's study on hyperelliptic sigma functions, J. Phys. A: Math. Gen. \textbf{34} (2001), 4721--4732. 
https://doi.org/10.1088/0305-4470/34/22/312

\bibitem{Nakayashiki-2010}
A. Nakayashiki, Sigma Function as A Tau Function,  Int. Math. Res. Not. \textbf{2010} (2010), 373--394. https://doi.org/10.1093/imrn/rnp135

\bibitem{Novikov}
S. P. Novikov, The periodic problem for the Korteweg-de Vries equation, Funct. Anal. Appl. \textbf{8} (1974), 236--246. https://doi.org/10.1007/BF01075697


\bibitem{Shiota-1986}
T. Shiota, Characterization of Jacobian varieties in terms of soliton equations, Invent. Math. \textbf{83} (1986), 333--382. https://doi.org/10.1007/BF01388967

\bibitem{L.-Wei-2011}
L. Wei, Multiple periodic-soliton solutions to Kadomtsev–Petviashvili equation, Appl. Math. Comput. \textbf{218} (2011), 368--375. 

https://doi.org/10.1016/j.amc.2011.05.072

\bibitem{WW}
E. T. Whittaker, G. N. Watson, A course of modern analysis, 5th ed., Cambridge University Press, Cambridge, 2021. https://doi.org/10.1017/9781009004091

\bibitem{Xu-Chen-Dai-2014}
Z. Xu, H. Chen, Z. Dai, Rogue wave for the $(2+1)$-dimensional Kadomtsev-Petviashvili equation, Appl. Math. Lett. \textbf{37} (2014), 34--38. 

https://doi.org/10.1016/j.aml.2014.05.005

\bibitem{Yang}
F. Yang, D. Qi, D. Y. Cheng, Z. H. Qing, Hyperelliptic Function Solutions of Three Genus for KP Equation Using Direct Method, Commun. Theor. Phys. \textbf{53} (2010), 615--618. https://doi.org/10.1088/0253-6102/53/4/05

\bibitem{Zakharov-Schulman-1988}
V. E. Zakharov, E. I. Schulman, On additional motion invariants of classical Hamiltonian wave systems, Physica D: Nonlinear Phenomena \textbf{29} (1988), 283--320. 

https://doi.org/10.1016/0167-2789(88)90033-4

\bibitem{Zakharov-Shabat-1974}
V. E. Zakharov, A. B. Shabat, A scheme for integrating the nonlinear equations of mathematical physics by the method of the inverse scattering problem. I, Funct. Anal. Appl. \textbf{8} (1974), 226--235. https://doi.org/10.1007/BF01075696

\bibitem{Zakharov-Shabat-1979}
V. E. Zakharov, A. B. Shabat, Integration of nonlinear equations of mathematical physics by the method of inverse scattering. II, Funct. Anal. Appl. \textbf{13} (1979), 166--174. https://doi.org/10.1007/BF01077483

\bibitem{Zhao-Fan-Luo-2016}
P. Zhao, E. Fan, L. Luo, Quasiperiodic solutions of the Kadomtsev-Petviashvili equation via the multidimensional Baker-Akhiezer function generated by the Broer-Kaup hierarchy, J.  Math. Anal. Appl. \textbf{435} (2016), 38--60. 
 
https://doi.org/10.1016/j.jmaa.2015.10.011


\end{thebibliography}



\end{document}